\documentclass[copyright,creativecommons]{eptcs}
\providecommand{\event}{GandALF 2014} 
\usepackage{amssymb, amsthm}
\usepackage{breakurl}             
\usepackage{mathtools}
\usepackage{wrapfig}
\usepackage{epsfig}
\usepackage{epstopdf}
\usepackage{pgf,tikz,paralist,pseudocode}
\usetikzlibrary{arrows,calc,shapes,snakes,automata,backgrounds,petri,positioning}

\newcommand{\false}{\mathsf{false}}
\newcommand{\true}{\mathsf{true}}

\newcommand{\nin}{\notin}
\newcommand{\ABBsim}{\mathsf{AB\oB\sim}}
\newcommand{\cM}{\mathcal{M}}

\newcommand{\vel}{\vee}
\newcommand{\et}{\wedge}
\newcommand{\sat}{\models}
\newcommand{\then}{\rightarrow}

\newcommand{\bM}{\mathbf{M}}
\newcommand{\cV}{\mathcal{V}}

\newcommand{\AABB}{\mathsf{A\bar{A}B\bar{B}}}
\newcommand{\AABBsim}{\mathsf{A\bar{A}B\bar{B}}\sim}

\newcommand{\ABB}{\mathsf{AB\bar{B}}}

\newcommand{\sP}{\cP}
\newcommand{\sC}{\cC}
\newcommand{\cS}{\mathcal{S}}
\newcommand{\cT}{\mathcal{T}}

\newcommand{\bG}{\mathbf{G}}

\newcommand{\duplicator}{{\diamond}}
\newcommand{\spoiler}{{\scriptscriptstyle \square}}
\newcommand{\diamondA}{\langle \mathsf{A} \rangle}
\newcommand{\diamondAbar}{\langle \mathsf{\bar{A}} \rangle}
\newcommand{\diamondB}{\langle \mathsf{B} \rangle}
\newcommand{\diamondBbar}{\langle \mathsf{\bar{B}} \rangle}
\newcommand{\diamondG}{\langle \mathsf{G} \rangle}

\newcommand{\boxA}{[\mathsf{A}]}

\newcommand{\boxB}{[\mathsf{B}]}

\newcommand{\boxG}{[\mathsf{G}]}

\newcommand{\sing}{{\pi}}
\newcommand{\unit}{{\mathsf{unit}}}

\newcommand{\boxCap}{[\cap]}
\newcommand{\diamondCap}{\langle \cap! \rangle}
\newcommand{\closure}{\mathsf{closure}}
\newcommand{\eclosure}{\mathsf{closure}^+}
\newcommand{\teclosure}{\mathsf{TF}^+}

\newcommand{\atoms}{\mathsf{atoms}}
\newcommand{\runs}{\mathsf{runs}}
\newcommand{\pruns}{\mathsf{\spoiler{-}pre}}
\newcommand{\spruns}{\mathsf{\spoiler{-}proj}}

\newcommand{\type}{\mathsf{type}}
\newcommand{\obs}{\mathsf{obs}}
\newcommand{\req}{\mathsf{req}}

\newcommand{\depA}{\stackrel{A}{\longrightarrow}}

\newcommand{\depBbar}{\stackrel{\overline{B}}{\longrightarrow}}

\newcommand{\cC}{\mathcal{C}}

\newcommand{\cP}{\mathcal{P}}

\newcommand{\bbD}{\mathbb{D}}
\newcommand{\bbI}{\mathbb{I}}
\newcommand{\bbN}{\mathbb{N}}

\newcommand{\cG}{\mathcal{G}}

\newcommand{\cL}{\mathcal{L}}

\newcommand{\cI}{\mathcal{I}}
\newcommand{\oA}{\overline{A}}

\newcommand{\oB}{\overline{B}}

\newcommand{\img}{\cI mg}

\renewcommand{\paragraph}[1]{\noindent {\bf #1.}~~}

\newtheorem{theorem}{Theorem}
\newtheorem{definition}{Definition}
\newtheorem{proposition}{Proposition}
\newtheorem{lemma}{Lemma}
\newtheorem{corollary}{Corollary}

\title{Interval-based Synthesis}
\author{
Angelo Montanari
\institute{Department of Mathematics and Computer Science\\ University of Udine, Italy}
\email{angelo.montanari@uniud.it}
\and
Pietro Sala 
\institute{Department of Computer Science\\ University of Verona, Italy}
\email{pietro.sala@univr.it}
}
\def\titlerunning{Interval-based Synthesis} 
\def\authorrunning{A. Montanari and P. Sala}
\begin{document}
\maketitle
\begin{abstract}

In this paper, we introduce the synthesis problem for Halpern and Shoham's interval temporal logic \cite{JACM::HalpernS1991} extended with an equivalence relation $\sim$ over time points
($HS\sim$ for short). In analogy to the case of monadic second-order logic of one successor \cite{buchiland69}, 
given an $HS\sim$ formula $\varphi$ and a finite set $\Sigma^T_{\spoiler}$ of proposition letters and temporal requests,
the problem consists of establishing whether or not, for all possible evaluations of elements in $\Sigma^T_{\spoiler}$ in every interval structure, there is an evaluation of the remaining proposition letters and temporal requests such that the resulting structure is a model for $\varphi$. We focus our attention on the decidability of the synthesis problem for some meaningful fragments of $HS\sim$, whose modalities are drawn from 
$\{A \ (meets), \ \bar{A} \ (met \ by), \ B \ (begun \ by), \ \bar{B} \  (begins) \}$, interpreted over finite linear orders and natural numbers. 
We prove that the synthesis problem for
$\ABBsim$ 
over finite linear orders is decidable (non-primitive recursive hard), while 
$\AABB$ turns out to be undecidable. In addition, we show that  if we replace finite linear orders by 
natural numbers, then the problem becomes undecidable even for $\ABB$. 
\end{abstract}
\section{Introduction}\label{sec:intro}

Since its original formulation by Church \cite{Church57}, the synthesis problem has received a lot
of attention in the computer science literature. A solution to the problem was provided by
B\"uchi  and Landweber in \cite{buchiland69}. In the last years, a number of extensions and
variants of the problem have been investigated, e.g., \cite{DBLP:journals/lmcs/Rabinovich07,RabThoCSL07}.
The synthesis problem for (point-based) temporal logic has been addressed in \cite{FMR12,ManWol84,PnuRos89}.

In this paper, we formally state the synthesis problem for interval temporal logic and present some basic 
results about it. We restrict ourselves to some meaningful fragments of Halpern and Shoham's 
modal logic of time intervals \cite{JACM::HalpernS1991} extended with an equivalence relation $\sim$ over 
time points ($HS\sim$ for short). 
The emerging picture is quite different from the one for the classical synthesis problem (for $MSO$).
In \cite{DBLP:journals/lmcs/Rabinovich07}, Rabinovich proves that the decidability of the monadic second-order theory of one successor $MSO(\omega, <)$ extended with a unary predicate $P$ ($MSO(\omega, <, P)$ for short) entails the decidability of its synthesis problem, that is, the synthesis problem for a monadic second-order theory is decidable if and only if its underlying theory is decidable. 
Here, we show that this is not the case with interval temporal logic. We focus our
attention on two fragments of $HS$, namely, the logic $\ABB$ of Allen's relations 
\emph{meets}, \emph{begun by}, and  \emph{begins}, and the logic $\AABB$ obtained
from $\ABB$ by adding a modality for the Allen relation \emph{met by}. In \cite{abb_natural},
Montanari et al.\ showed that the satisfiability problem for $\ABB$ over finite linear orders and
the natural numbers is $EXPSPACE$-$complete$, while, in \cite{abba_finite}, Montanari, Puppis,
and Sala proved that the satisfiability problem for $\AABB$ over finite linear orders is decidable, 
but not primitive recursive (and undecidable over the natural numbers). In this paper, we prove that the
synthesis problem for $\ABB$ over the natural numbers and for $\AABB$ over finite linear
orders turns out to be undecidable. Moreover, we show there is a significant blow up in 
computational complexity moving from the satisfiability to the synthesis problem  for $\ABB$ 
over finite linear orders: while the former is $EXPSPACE$-$complete$, the latter is 
$NON$-$PRIMITIVE$ $RECURSIVE$-$hard$. 
As a matter of fact, such an increase in the complexity is paired with an increase in the expressive 
power of the logic: one can exploit universally quantified variables, that is, propositional letters
under the control of the environment,
to constrain the length of intervals in a way which is not allowed by $\ABB$.

The rest of the paper is organized as follows. In Section \ref{sec:logic}, we introduce syntax and semantics of the logic $\AABB\sim$ and its fragments. In Section \ref{sec:synthdef}, we define the synthesis problem for interval temporal logic, focusing our attention on the considered fragments. The problem is then systematically investigated in the next two sections, where decidable and undecidable instances are identified (a summary of the results is given in Table \ref{table:undecidabilityandcomplexity}). Conclusions provide an assessment of the work and outline future research directions.

\section{The logic $\AABB\sim$ and its fragments}\label{sec:logic}

In this section, we provide syntax and semantics of the fragments
of 
$HS\sim$ we are interested in.
The maximal  fragment that we take into consideration is $\AABB\sim$, which features 
unary modalities $\diamondA$, $\diamondAbar$, $\diamondB$, and $\diamondBbar$ for 
Allen's binary ordering relations \emph{meets}, \emph{met by}, \emph{begun by}, and 
\emph{begins} \cite{allen83}, respectively, plus a special proposition letter $\sim$, to be interpreted
as an equivalence relation. The other relevant fragments are $\AABB$, $\ABBsim$, and $\ABB$.


Formally, let $\Sigma$ be a set of proposition letters, with $\sim \in \Sigma$. Formulas of $\AABB\sim$
are built up from proposition letters in $\Sigma$ by using Boolean connectives $\vel$ and $\neg$,
and unary modalities from the set  $\{\diamondA, \diamondAbar, \diamondB, \diamondBbar\}$. 
Formulas of the fragments $\AABB$, $\ABBsim$, and $\ABB$ are defined in a similar way.
We will often make use of shorthands like $\varphi_1\et\varphi_2 \,=\, \neg(\neg\varphi_1 \vel \neg\varphi_2)$, 
$\boxA\varphi \,=\, \neg\diamondA\neg\varphi$, $\boxB\varphi \,=\, \neg\diamondB\neg\varphi$, 
$\true \,=\, a \vee \neg a$, and $\false \,=\, a \wedge \neg a$, for some $a \in \Sigma$

As for the semantics, let $\bbD=(D,<)$ be a linear order, called \emph{temporal domain}. We denote by $\bbI_\bbD$ the set of all closed intervals $[x,y]$ over $\bbD$, with $x=y$ or $x<y$, abbreviated  $x \leq y$ (non-strict semantics).
%
We call \emph{interval structure} any Kripke structure of the form $\bM=(\bbD, A, \bar{A}, B, \bar{B},\cV)$. 
$\cV:\bbI_\bbD\then\sP(\Sigma)$ is a function mapping intervals to sets of proposition letters.
$A, \bar{A}, B$,  and $\bar{B}$ denote Allen's relations ``meet'', ``met by'', ``begun by'', and 
``begins'', respectively, and are defined as follows:
$[x,y] \mathrel{A} [x',y'] \text{ iff } y=x'$, 
$[x,y] \mathrel{\bar{A}} [x',y'] \text{ iff } x=y'$,
$[x,y] \mathrel{B} [x',y']  \text{ iff } x=x' \et y'<y$, and 
$[x,y] \mathrel{\bar{B}} [x',y'] \text{ iff } x=x' \et y<y'$.
For the sake of brevity, in the following we will write $\bM=(\bbD,\cV)$ for $\bM=(\bbD, A, \bar{A}, B, \bar{B},\cV)$.

Formulas are interpreted over an interval structure $\bM=(\bbD,\cV)$ and an initial interval $I\in \bbI_\bbD$ as follows:
$\bM,I \sat a$ iff $a\in\cV(I)$, $\bM,I \sat \neg\varphi$ iff $\bM,I \not\sat\varphi$,
$\bM,I \sat \varphi_1 \vel \varphi_2$ iff $\bM,I \sat \varphi_1$ or $\bM,I \sat \varphi_2$,
and, for all $R\in\{A,\bar{A},B,\bar{B}\}$,
$$
  \bM,I \sat \langle \mathsf{R} \rangle \varphi
  \qquad\text{iff}\qquad 
  \text{there exists $J\in\bbI_\bbD$ such that $I\mathrel{R} J$ and $\bM,J \sat \varphi$.}
$$
The special proposition letter $\sim$ is interpreted as an equivalence relation over $\bbD$, that is, 
(i) $x\sim x $ for all $x\in D$, (ii) forall $x, y \in D$, if $x\sim y$, then $y\sim x$, and for all $x, y, z\in D$,
if $x\sim y$ and $y\sim z$, then $x\sim z$. Now, for all  $x, y \in D$, with $x\leq y$, $\bM,[x,y]\sat \sim$ if
(and only if) $x\sim y$.
%
%
In the following, we will write $x\sim y$ for  $\bM,[x,y]\sat \sim$ whenever the context, that is, 
the pair $(\bM,[x,y])$, is not ambiguous.

We say that a formula $\varphi$ is \emph{satisfiable} over a 
class $\sC$ of interval structures if $\bM,I\sat\varphi$ for some 
$\bM=(\bbD,\cV)$ in $\sC$ and some interval 
$I\in\bbI_\bbD$.
In the following, we restrict our attention to the class $\sC_{fin}$ of finite linear orders and to (the class
$\sC_{\bbN}$ of linear orders isomorphic to) $\bbN$.
Without loss of generality (we can always suitably rewrite $\varphi$), we assume
the initial interval on which $\varphi$ holds (in a model for it) to be the interval $[0,0]$.

In the following, we will often make use of the following formulas.
The formula $\boxB\false$ (hereafter, abbreviated $\sing$) holds over all and only the \emph{singleton} intervals $[x,x]$. Similarly, the formula $\boxB\boxB\false$ (abbreviated $\unit$) holds over the \emph{unit-length} intervals over a discrete order, e.g., over the intervals of $\bbN$ of the form $[x,x+1]$. Finally, the formula $\boxA\boxA\varphi$ ($\boxG\varphi$ for short), interpreted over the initial interval $[0,0]$,
forces $\varphi$ to hold universally, that is, over all intervals.
For the sake of readability, from now on, we will denote by $\Sigma$ the set 
of all and only those proposition letters that appear in the formula $\varphi$ under
consideration, thus avoiding  tedious parametrization like $\Sigma(\varphi)$
(it immediately follows that $\Sigma$ is always assumed to be a finite set of proposition
letters). 

Given an $\AABBsim$ formula $\varphi$, we define its  \emph{closure} as the set $\closure(\varphi)$ 
of all its sub-formulas and all their negations (we identify $\neg\neg\psi$ with $\psi$, $\neg\diamondA\psi$ with $\boxA\neg\psi$, and so on). 
For a technical reason that will be clear soon, we also introduce the \emph{extended closure} of $\varphi$, 
denoted by $\eclosure(\varphi)$, that extends $\closure(\varphi)$ by adding all formulas 
of the form $\langle \mathsf{R} \rangle \psi$ and $[R]\psi$, for $R\in\{A, B, \bar{A}, \bar{B}\}$ 
and $\psi\in\closure(\varphi)$. Moreover, we denote by $\teclosure(\varphi)\subseteq \eclosure(\varphi)$ 
the set 
$\{\langle R\rangle \psi \in \eclosure(\varphi):  R\in \{A, \oA, B, \oB\}
\}$. From now on, given an $\AABBsim$ formula $\varphi$, we will denote
by $\Sigma^T$ the set $\Sigma \cup TF^+(\varphi)$.

\section{The synthesis problem for interval temporal logic}\label{sec:synthdef}

We are now ready to define the synthesis problem for the interval logic $\AABBsim$ (the
definition immediately transfers to all its fragments) with respect to the class
of finite linear orders and to (any linear order isomorphic to) $\bbN$. Without loss 
of generality, we will refer to a linear order which is either $\bbN$ or one of its finite prefixes. 
To start with, we introduce the notion of admissible run.
\begin{definition}\label{def:admrun}
Let $\varphi$ be an $\AABBsim$ formula  and let $\Sigma^{T}_{\spoiler}\subseteq \Sigma^T$.
An \emph{admissible run} $\rho$ on the pair $(\varphi,\Sigma^{T}_{\spoiler})$
is a finite or infinite sequence  of pairs $\rho=([x_0,y_0],\sigma_0)([x_1,y_1],$ $ \sigma_1)\ldots$ 
such that:
\begin{enumerate}

\item if $\rho$ is finite, that is, $\rho=([x_0,y_0],\sigma_0)\ldots([x_n,y_n], \sigma_n)$, then there 
exists $m>0$ such that $n= 2 \cdot m \cdot (m+1)$ and, for each $[x,y]\in \bbI(\{0, \ldots, m-1\})$, there 
exists $0\leq i\leq n$ such that $[x,y]=[x_i,y_i]$, while if $\rho$ is infinite, then, for every $[x,y]\in 
\bbI(\bbN)$, there exists 
$i \geq 0$ such that $[x,y]=[x_i,y_i]$;

\item $[x_0,y_0]=[0,0]$, for all $0 < i \ (\leq n)$,  $[x_i,y_i]\in \bbI(\bbN)$,  and for every even index $i$,
$[x_i,y_i]=[x_{i+1}, y_{i+1}]$, $\sigma_i\subseteq\Sigma^T_{\spoiler}$ (the set of proposition letters and
temporal requests in $\Sigma^T_{\spoiler}$ true on $[x_i,y_i]$), $\sigma_{i+1} \subseteq \Sigma^T \setminus \Sigma^T_{\spoiler}$ (the set of proposition letters and temporal requests in $\Sigma^T \setminus 
\Sigma^T_{\spoiler}$ true on $[x_i,y_i]$), and for all $j$, with $j \neq i$ and $j \neq i+1$, $[x_i,y_i]\neq [x_j,y_j]$; 

\item if $y_{i+1}\neq y_i$, then $y_{i+1}=y_i+1$ and for all $[x,y]\in \bbI(\bbN)$, with $y<y_{i+1}$, there exists
$0\leq j<i+1$ such that $[x_j,y_j]=[x,y]$.  

\end{enumerate}
\end{definition}
Conditions $1$-$3$ define the rules of a possibly infinite game between two players $\spoiler$ (spoiler) and  $\duplicator$  (duplicator), which are responsible of the truth values of proposition letters and temporal requests
in $\Sigma^T_{\spoiler}$ and $\Sigma^T\setminus \Sigma^T_{\spoiler}$, respectively.
The game can be informally described as follows. At the beginning, $\spoiler$ chooses an interval 
$[x,y]$ and defines his labeling for $[x,y]$; $\duplicator$ replies to $\spoiler$ by defining her 
labeling for $[x,y]$ as required by condition 2.
In general, $\spoiler$  makes his moves at all even indexes, while $\duplicator$ executes her moves at all odd indexes 
by completing the labeling of the interval chosen by $\spoiler$ at the previous step. Condition $1$ guarantees that all  intervals on $\bbN$ (infinite case) or on a finite prefix of it (finite case) are visited by the play. Condition  $2$ forces every visited interval to be visited exactly once.  Condition $3$ imposes an order according to which intervals are visited. More
precisely,  condition $3$ prevents $\spoiler$ from choosing an interval $[x,y]$ before he has visited all intervals $[x'y']$, with $y'<y$, that is, $\spoiler$ cannot jump ahead along the time domain without first defining the labeling of all intervals ending at the points he would like to cross.
 
Let $\runs(\varphi, \Sigma^T_{\spoiler})$  be the set of of all possible admissible runs on the pair $(\varphi,\Sigma^{T}_{\spoiler})$. We denote by $\pruns(\varphi, \Sigma^T_{\spoiler})$ the set of all odd-length finite 
prefixes of admissible runs in $\runs(\varphi, \Sigma^T_{\spoiler})$, that is, prefixes in which the last move
was done by $\spoiler$, and by $\spruns(\varphi, \Sigma^T_{\spoiler})$ the set of all infinite subsequences 
of admissible runs in $\runs(\varphi, \Sigma^T_{\spoiler})$ that contain all and only the pairs occurring at 
even positions (formally, $\spruns(\varphi, \Sigma^T_{\spoiler})=\{ \rho'=([x_0,y_0], \sigma_0)([x_1,y_1], 
\sigma_2)\ldots: \exists \rho \in \runs(\varphi, \Sigma^T_{\spoiler}) \text{ such that } \forall i  ( \rho[2i]=\rho'[i])   \}$).
%
It can be easily seen that an admissible run $\rho$ provides a labeling $\cV$ for some 
candidate model of $\varphi$ by enumerating all its intervals $[x,y]$ following the order of their right 
endpoints $y$, that is,
for any point $y$, intervals of the form $[x,y]$  may appear in $\rho$ shuffled in an arbitrary order, which depends
on  the choices of $\spoiler$, but if  $\rho$ features a labelled interval $[x,y+1]$, then all labeled intervals $[x',y]$,  
with $0\leq x'\leq y$, must  occur in $\rho$ before it.
For any $\rho=([x_0,y_0], \sigma_0)([x_1,y_1], \sigma_1)\ldots$ in $\runs(\varphi, \Sigma^T_{\spoiler})$,
we denote by $\rho_I$ and by $\rho_{\sigma}$ the sequence $[x_0,y_0] [x_1,y_1] \ldots $ and the  sequence $\sigma_0 \sigma_1\ldots$ obtained by projecting $\rho$ on its first component and its second component, respectively.

Any admissible run $\rho$ on 
$(\varphi,\Sigma^{T}_{\spoiler})$ induces an interval structure $\bM_{\rho}=(\bbD, \cV)$,
called \emph{induced structure},
where $\bbD=\bbN$, if $\rho$ is infinite, or $\bbD=\{ 0< \ldots< m-1 \}$, with $|\rho|=2 \cdot m \cdot(m+1)$ (such an $m$ 
exists by definition of finite admissible run) otherwise, and $\cV([x,y])=(\rho[i]\cap \Sigma) \cup (\rho[i+1]\cap\Sigma)$,
where $i$ is even and $\rho_I[i]=[x,y]$.
Both the existence and the uniqueness of such an index $i$ are guaranteed by condition 2 of the definition of admissible run,  and thus the function $\cV$ is correctly defined. In particular, we can define a (unique) bijection $f_{\rho}:\bbI(\bbD)\rightarrow \bbN$ such that, for each $[x,y]\in \bbI(\bbD)$, $f_{\rho}([x,y])$ is even and $\rho_I[f_{\rho}([x,y])]=[x,y]$.

We say that an admissible run $\rho$ is \emph{successful} if and only if 
$\bM_{\rho} (=(\bbD, \cV)),[0,0]\models \varphi$ and for all $[x,y] \in \bbI(\bbD)$ and  $\psi\in \teclosure(\varphi)$, it holds that $\bM_\rho,[x,y]\models \psi$ iff $\psi \in \rho_\sigma[f_{\rho}([x,y])]\cup \rho_\sigma[f_{\rho}([x,y])+1]$. 
Given  a pair $(\varphi,\Sigma^{T}_{\spoiler})$, a \emph{$\Sigma^{T}_{\spoiler}$-response strategy} is a function $S_{\diamond}:\pruns(\varphi,\Sigma^{T}_{\spoiler})\rightarrow \cP(\Sigma^T\setminus\Sigma^{T}_{\spoiler})$. 
Moreover, given an infinite $\Sigma^{T}_{\spoiler}$-response strategy $S_{\diamond}$ and a sequence $\rho_{\spoiler}=([x_0,y_0], \sigma_0)([x_1,y_1], \sigma_1)\ldots$ in $\spruns(\varphi,\Sigma^{T}_{\spoiler})$, 
we define the \emph{response of $S_{\diamond}$ to $\rho_{\spoiler}$} as the infinite admissible run $\rho=([x'_0,y'_0], \sigma'_0)([x'_0,y'_0],\sigma'_0) \ldots$, where, for all $i\in \bbN$, $\rho[2i]=\rho_{\spoiler}[i]$ 
and $\rho[2i+1]=(\rho_{I}[2i], S_{\diamond}(\rho[0\ldots 2i]))$. 

The \emph{finite synthesis} problem for $\AABBsim$, that is, the winning condition for $\duplicator$ on the game defined by conditions $1$-$3$, can be formulated as follows.

\begin{definition}\label{synthesis}
Let  $\varphi$  be an $\AABBsim$ formula and $\Sigma^{T}_{\spoiler}\subseteq \Sigma^T$. We say that the pair $(\varphi,\Sigma^{T}_{\spoiler})$ admits a \emph{finite synthesis} if and only if there exists a $\Sigma^{T}_{\spoiler}$-response strategy $S_{\diamond}$ such that for every $\rho_{\spoiler} \in \spruns(\varphi,\Sigma^{T}_{\spoiler})$, the response $\rho$ of $S_{\diamond}$ to $\rho_{\spoiler}$ has a finite prefix $\rho[0\ldots n]$ which is a successful admissible run.
\end{definition}
By definition of (finite) admissible run, there exists $m$ such that $n=2 \cdot m \cdot (m+1)$.
Basically, when the labeling is completed  for the intervals ending at some point $y' \leq y$
and  $\bM_{\rho[0\ldots 2 \cdot (y+1) \cdot (y+2)]}$ is a model for $\varphi$, then $\duplicator$ 
wins and she can safely ignore the rest of the run $\rho$ (as it happens with reachability games).
To generalize the above definition to the $\bbN$-synthesis problem, it suffices to drop the prefix 
condition of Definition \ref{synthesis} and to constrain $\rho$ to be a successful admissible run. 
In general, we say that a pair $(\varphi,\Sigma^{T}_{\spoiler})$ is a \emph{positive instance} of the finite synthesis
(resp., $\bbN$-synthesis) problem if and only if it admits a finite synthesis (resp.,  $\bbN$-synthesis). 
A $\Sigma^{T}_{\spoiler}$-response strategy $S_{\diamond}$, which witnesses that  $(\varphi,\Sigma^{T}_{\spoiler})$
is a positive instance of the finite synthesis (resp., $\bbN$-synthesis) problem, is called a \emph{winning strategy}.

We conclude the section by showing how to exploit the finite synthesis problem to express in $\ABB$ 
(the smallest fragment we consider in this work) a temporal property that can be expressed neither in $\ABB$ nor in $\AABB$ (in the usual satisfiability setting).
While there is a common understanding of what is meant by enforcing a property on a model via satisfiability, 
such a notion has various interpretations in the synthesis framework.
We assume the following interpretation: forcing a property $P$ on a model 
means requiring that for all $\Sigma^{T}_{\spoiler}$-response winning strategies $S_{\diamond}$, there exists 
a sequence $\rho_{\spoiler} \in \spruns(\varphi, \Sigma^T_{\spoiler})$ such that the run $\rho$, which is the 
response of $S_{\duplicator}$ to $\rho_{\spoiler}$, features the property $P$ on all possible models 
$\bM_{\rho[0,\ldots n]}$, where $\rho[0,\ldots n]$ is a successful admissible run. This amounts to say that no matter how $\duplicator$ plays, if she wants to win, then there is always a choice for $\spoiler$ that constrains  property $P$ to hold on the model that is built at the end of the play (that is, $\duplicator$ cannot win avoiding property $P$).

Let us consider, for instance, the following property:
\emph{there exists at least one occurrence of an event (of type) $e_1$, each occurrence of $e_1$ is followed by an occurrence of an event (of type) $e_2$, occurrrences of $e_1$ are disjoint, occurrences of $e_2$ are disjoint, occurrences of $e_1$ and $e_2$ are disjoint,  and for every two consecutive occurrences of $e_1$ and $e_2$, the duration of the occurrence of $e_2$ is greater than or equal to the duration of the occurrence of $e_1$.} \\
\begin{wrapfigure}{r}{0.45\textwidth}
\vspace{-0.5cm}
\centering
\includegraphics{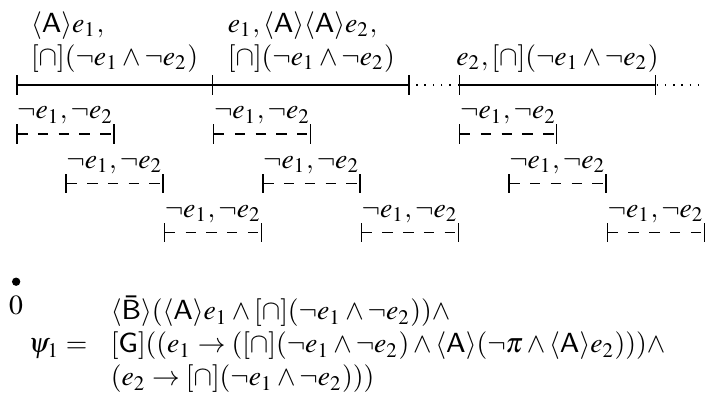}
\vspace{-1cm}

\end{wrapfigure}
%
In the following, we specify the input  $(\varphi,\{corr_{\spoiler}\})$ of a synthesis problem, where $\varphi$ is defined as the conjunction $\psi_0 \wedge \psi_1\wedge \psi^1_2\wedge\psi^2_2\wedge \psi_3 \wedge \psi_4$ and $\spoiler$ controls the proposition letter $corr_{\spoiler}$ only.
To simplify the encoding, we will make use two auxiliary modalities $\boxCap\psi =\boxB\boxA \psi \wedge \boxB \psi$  
and $\diamondCap \psi =\boxB\neg \psi \wedge \diamondB\diamondA \psi \wedge \boxB( \diamondA \psi \rightarrow \boxB\boxA\neg \psi)$. By definition, $\boxCap\psi$ holds on an interval $[x,y]$ if $\psi$ holds on all intervals beginning at some $z$, with $x\leq z<y$, and different from $[x,y]$, while $\diamondCap\psi$ holds on $[x,y]$ if there exists one and only one interval beginning at some $z$, with $x < z<y$, on which $\psi$ holds. 
For the sake of simplicity, we constrain $e_1$ and $e_2$ to hold only over intervals with a duration by means of the formula $\boxG((e_1 \vee e_2) \rightarrow \neg \pi)$ (formula $\psi_0$).

Formula $\psi_1$ takes care of  the initial condition (there exists at least one occurrence of $e_1$)  and of the relationships between $e_1$- and $e_2$-labeled intervals. Its first conjuct forces the first event to be $e_1$. The second one (whose outermost operator is $\boxG$) constrains events $e_1$ to be pairwise disjoint and disjoint from events $e_2$, forces each $e_1$-labeled interval to be followed by an $e_2$-labeled one, and constrains events $e_2$ to be pairwise disjoint and disjoint from events $e_1$.

Formula $\psi^1_2$ forces an auxiliary proposition letter $end_1$ to hold only at the right endpoint of $e_1$-labeled intervals. 
%
The first conjunct forces $\neg end_1$ to hold on all intervals that preceeds the first occurrence of an $e_1$-labeled interval. The second one (whose outermost operator is $\boxG$) forces every $e_1$-labeled interval to meet an 
$end_1$-labeled interval and prevents $end_1$ labeled-intervals from occurring inside an $e_1$-labeled interval.
Moreover, it forces $end_1$ to hold on point intervals only and constrains all point-intervals (but the first one)
that  belong to an interval that connects two consecutive $e_1$-labeled intervals to satisfy $\neg end_1$.
Formula $\psi^2_2$ imposes the very same conditions on proposition letter $end_2$ with respect to 
$e_2$-labelled intervals, and it can be obtained from $\psi^1_2$ by replacing $end_1$ by $end_2$  
and $e_1$ by  $e_2$.\\
%
%
\begin{figure}[h!]
\centering
\includegraphics{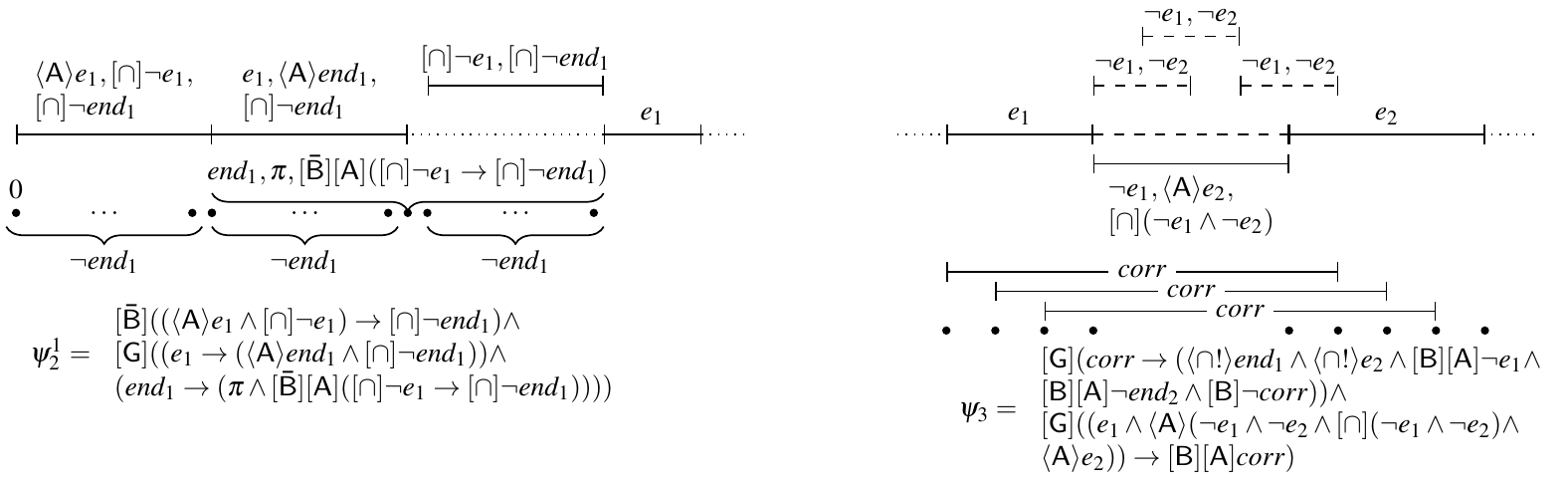}
\end{figure}

Formula $\psi_3$ makes use of the proposition letter $corr$ to establish a correspondence between consecutive
$e_1$- and $e_2$-labeled intervals, that is, $corr$ maps points belonging to an  $e_1$-labeled interval $[x,y]$ 
to points belonging to an $e_2$-labeled $[x',y']$ if and only if there is no point $y<y''<y'$ that begins an $e_1$- or an $e_2$-labeled interval.
The first conjunct (whose outermost operator is $\boxG$) constrains every $corr$-labeled interval 
to cross exactly one point labeled with $end_1$ and to include the starting point of exactly one
$e_2$-labeled  interval. 
Moreover, it prevents $e_1$-labeled intervals to begin and $e_2$-labeled interval to end
at a point belonging to a $corr$-labeled interval. Finally, it allows at most one  $corr$-labeled interval 
to start at any given point.
The second conjunct (whose outermost operator is $\boxG$) forces a $corr$-labeled interval
to start at any point belonging to an $e_1$-labeled interval $[x,y]$ which has an  $e_2$-labeled 
interval as its next $e_i$-labeled interval, with $i\in \{1,2\}$.
All in all, it constrains $corr$-labeled intervals to connect all points belonging to the $e_1$-labeled
interval  to the points belonging to the next $e_2$-labeled interval.  Since exactly one $corr$-labeled 
interval can start at each point in $e_1$, $corr$ can be viewed as a function from points in $e_1$ to 
points in $e_2$.
To capture the intended property, however, we further need to force such a function to be injective. 
This is done by formula $\psi_4$, which exploits the interplay between $\spoiler$ and $\duplicator$.
\begin{wrapfigure}{r}{0.43
\textwidth}
\centering
\includegraphics{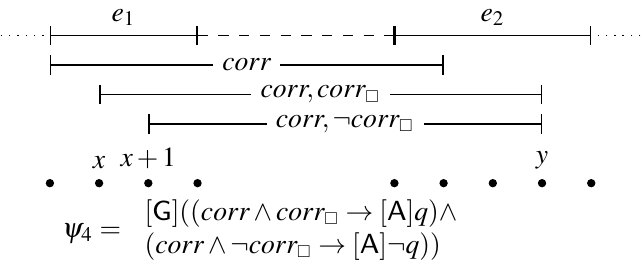}
%
%
%
%
%
%
\end{wrapfigure}

It is worth pointing out that if we take a look at formula $\psi_4$ from the point of view of the
satisfiability problem, it does not add any constraint to the proposition letter $corr$. 
Indeed, $\psi_4$ can be trivially satisfied by forcing $corr_{\spoiler}$ and $q$ to be
always true or always false in the model.
In the context of the (finite) synthesis problem, things are different:
$\duplicator$ has no control on the proposition letter $corr_{\spoiler}$,
and if she tries to violate injectivity of $corr$ (as depicted in the graphical account for $\psi_4$),
then $\spoiler$ has a strategy to win, as shown by the following (portion of a) run:
\vspace{-0.3cm}
\begin{center}
\includegraphics{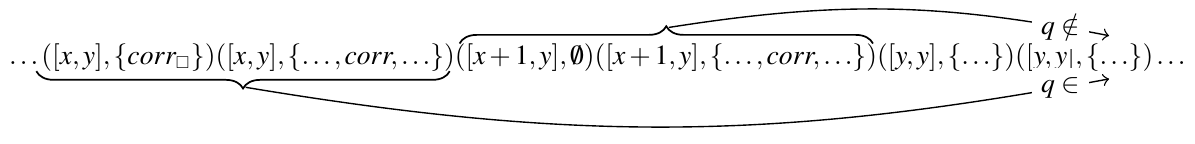}
\end{center}
\vspace{-0.3cm}
%
In general, $\spoiler$ has a strategy to impose that  for each point $y$, there exists at most one $corr$-labeled interval.
Suppose that, at a certain position of the run, $\rho$ $\spoiler$ and $\duplicator$ are playing on the labeling of all intervals ending at a given point $y$. For each $0\leq x\leq y$, $\spoiler $, who always plays first, may choose any value for $corr_{\spoiler}$ on $[x,y]$ until $\duplicator$ chooses to put $corr$ in her reply to a move of $\spoiler$ on an interval $[x',y]$,  for some $x'$.
From that point on, for all intervals $[x'',y]$ which have not been labeled yet, if  $\spoiler$ has put $corr_{\spoiler}$ on 
$[x',y]$, he will label $[x'',y]$ with $\neg corr_{\spoiler}$, and if he has put  $\neg corr_{\spoiler}$ on $[x',y]$, he will label $[x'',y]$ with $ corr_{\spoiler}$.

\section{The big picture: undecidability and complexity reductions}\label{section:undecidability}

\begin{table}
\footnotesize
\centering
\begin{tabular}{cccc}
Logic & Linear Order & Satisfiability & Synthesis \\
\hline
\hline\\[-2ex] 
$\ABB$ & Finite & \begin{tabular}{c} Decidable \cite{abb_natural}  (EXPSPACE-complete) \end{tabular} &
\begin{tabular}{c} Decidable (\textbf{NonPrimitiveRecursive-hard}) \end{tabular}  \\
\hline
$\ABB$ & $\bbN$ & \begin{tabular}{c} Decidable \cite{abb_natural}
\\ (EXPSPACE-complete)  \end{tabular} & \textbf{Undecidable}\\
\hline
$\ABB\sim$ & Finite &\begin{tabular}{c} Decidable \cite{montanari2013adding} \\ (NonPrimitiveRecursive-hard)  \\\end{tabular} & 
\begin{tabular}{c} 
\textbf{Decidable} \\ (NonPrimitiveRecursive-hard)
\end{tabular} \\
\hline\\[-2ex] 
$\ABB\sim$ & $\bbN$ & Undecidable \cite{montanari2013adding} & 
Undecidable \\
\hline
$\AABB$ & Finite &\begin{tabular}{c} Decidable \cite{abba_finite} \\ (NonPrimitiveRecursive-hard)  
\\
\end{tabular} & 
\textbf{Undecidable} \\
\hline\\[-2ex] 
$\AABB$ & $\bbN$ & Undecidable \cite{abba_finite} & 
Undecidable \\
\hline\\[-2ex] 
$\AABB\sim$ & Finite & Undecidable \cite{montanari2013adding} & 
Undecidable \\
\hline\\[-2ex] 
$\AABB\sim$ & $\bbN$ & Undecidable & 
Undecidable \\
\hline
\end{tabular} 
\caption{Decidability and complexity of the satisfiability and synthesis problems for the considered $HS$ fragments. Results written in bold are given in the present work; results with no explicit reference immediately follow from those given
in this paper or in other referred ones.}
\label{table:undecidabilityandcomplexity}
\vspace{-0.5cm}
\end{table}

In this section, we state (un)decidability and complexity results for the synthesis problem for the three fragments $\ABB, \ABBsim$, and $\AABB$. We consider both finite linear orders and the natural numbers.
The outcomes of such an analysis are summarized in Table \ref{table:undecidabilityandcomplexity}.
All the reductions we are going to define make use of (Minsky) counter machines or of their lossy variants.
A \emph{ (Minsky) counter machine} \cite{Minsky:1967:CFI:1095587} is a triple 
$\cM=(Q,k,\delta)$, where $Q$ is a finite set of 
states, $k$ is 
the number of counters, whose values range over $\bbN$, and $\delta$ is 
a function that maps each state $q\in Q$ to a transition rule having one 
of the following forms:
\begin{itemize}
  \item $\texttt{inc}(i) \texttt{ and goto}(q')$, where $i\in\{1,...,k\}$ is a 
        counter and $q'\in Q$ is a state: 
        whenever $\cM$ is in state $q$, then it first increments the value
        of counter $i$ and then it moves to state $q'$;
  \item $\texttt{if}\; i=0 \texttt{ then goto}(q') \texttt{ else dec}(i)\texttt{ and goto}(q'')$, 
        where $i\in\{1,...,k\}$ is a counter and $q',q''\in Q$ are states: 
        whenever $\cM$ is in state $q$ and the value of the counter $i$ is equal to $0$ 
        (resp., greater than $0$), then 
        $\cM$ moves to state $q'$ (resp., it decrements the value of $i$ and moves to state $q''$).
\end{itemize}
A computation of $\cM$ is any sequence of configurations that conforms to the semantics of 
the transition relation. In the following, we will define and exploit a suitable reduction from the problem 
of deciding, given a  counter machine $\cM=(Q,k,\delta)$ and a pair of control states $q_{\mathsf{init}}$ and
$q_{\mathsf{halt}}$, whether or not every computation of $\cM$ that starts at state $q_{\mathsf{init}}$, 
with all counters initialized to $0$, eventually reaches the state $q_{\mathsf{halt}}$, with the values
of all counters equal to $0$ (\emph{$0$-$0$ reachability} problem). 
\begin{theorem}\label{thm:minsky}  (\cite{Minsky:1967:CFI:1095587})
The $0$-$0$ reachability problem for counter machines is undecidable.
\end{theorem}
If, given a configuration $(q,\bar{z})\in Q\times\bbN^k$, we allow a counter machine $\cM$
to non-deterministically execute an internal (lossy) transition and to move to a configuration 
$(q,\bar{z}')$, with $\bar{z}'\le\bar{z}$ (the relation $\le$ is defined component-wise 
on the values of the counters), we obtain a lossy counter machine.  
\begin{theorem}\label{thm:lossy}  (\cite{cheat_sheet})
The $0$-$0$ reachability problem for lossy counter machines is decidable
with NonPrimitive Recursive-hard complexity.
\end{theorem}



Notice that, given a counter machine $\cM$, a computation where lossy transitions have been never executed, namely, a \emph{perfect computation}, is a lossy computation, while, in general, a lossy computation cannot be turned into a computation which does not execure any lossy transition.

\begin{wrapfigure}{r}{0.6\textwidth}
\vspace{-0.7cm}
\hspace{-0.3cm} 
\includegraphics{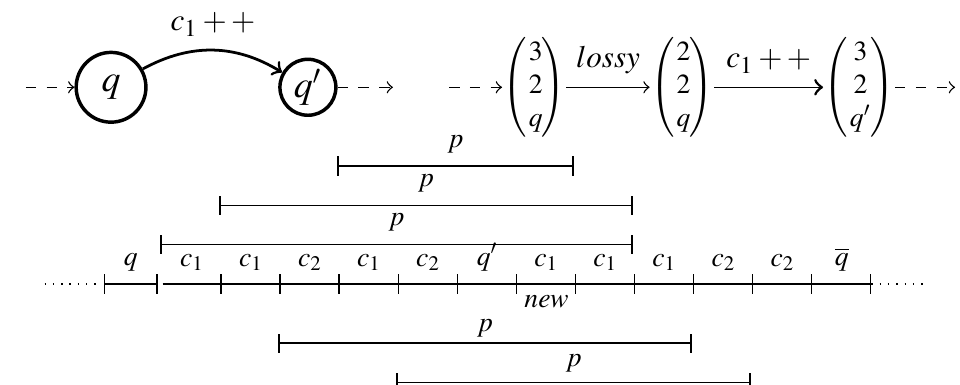}

\caption{\hspace{-0.3cm} Encoding of a lossy computation  in $\AABB$: incrementing states.}\label{fig:forwardinc}
\vspace{-0.5cm}
\end{wrapfigure}

Since lossy transitions are not under the control of the machine $\cM$ and they may take place at each
state of the computation, lossy computations and perfect computations can be viewed as two particular 
semantics for the computations of the same machine, the former being more relaxed (that is, it allows, in
general, a larger number of successful computations) than the latter.
Now we prove that the synthesis problem for $\AABB$ over finite linear orders is undecidable.
We elaborate on a result given by Montanari et al.\ in  \cite{abba_finite}, where, for any  
counter machine $\cM$, a formula $\varphi_{\cM}$ is given such that $\varphi_{\cM}$ is 
satisfiable over finite linear orders if and only if the corresponding counter machine $\cM$
has a $0$-$0$ lossy computation for two given states $q_0$ and $q_f$.
The idea is to encode the successful computation in a model for  $\varphi_{\cM}$. 
In the following, we will first briefly recall the key ingredients of such an encoding; then, we 
will show how to extend $\varphi_{\cM}$ with an additional formula that actually introduces a
new constraint only in the finite synthesis setting.
We start with a short explanation of how the basic features of a (candidate) model for $\varphi_{\cM}$ 
can be enforced. To help the reader, we provide a graphical account of the technique (interval structure 
in Figure \ref{fig:forwardinc}). Each configuration is encoded by means of a sequence of consecutive 
unit intervals. The first unit interval of any such sequence is labeled with a propositional letter $q_i$, 
where $q_i$ is a state of $\cM$. A unary encoding of the values of the counters is then provided by
making use of the unit intervals in between (the unit interval labeled with) $q_i$ and the next unit 
interval labeled with a state of $\cM$, say, $q_j$. Any such unit interval is labeled by exactly one proposition
letter $c_i$, with $i \in \{1,\ldots, k\}$. More precisely,  for all $i \in \{1,\ldots, k\}$, the value of the 
counter $c_i$ in the configuration beginning at $q_i$ is given by the number of $c_i$-labelled unit 
intervals between $q_i$ and $q_j$. 
We show now how to encode the two kinds of transition of $\cM$.  Let $q$ be the current
state. We first consider the case of increasing transitions of the form $\texttt{inc}(i) \texttt{ and 
goto}(q')$. It is easy to write a formula that forces the unit interval labeled with a state of $\cM$ next 
to $q$ to be labeled with $q'$. It is also easy to force exactly one $c_i$-labeled unit interval 
in the next configuration 
to be labeled with a special proposition letter $new$, to identify the $c_i$-labeled interval just 
introduced to mimic the increment of the counter $i$. 
Let us consider now transitions of the form  $\texttt{if}\; i=0 \texttt{ then goto}(q') \texttt{ else dec}(i)
\texttt{ and goto}(q'')$. We first verify whether there are not $c_i$-labeled unit intervals in the 
current configuration by checking if the formula $\boxB \boxA \neg c_i$ 
holds over the interval that begins at the left endpoint of (the unit interval labeled with) $q$
and ends at the left endpoint of the next unit interval labeled with a state of $\cM$. If such a formula 
does not hold, we have to mimic the decreasing of the counter $c_i$ by one. To this end,
we introduce another special proposition letter $del$ and we force it to hold over one of the
$c_i$-labeled intervals of the current configuration. Intuitively, the interval marked by $del$ is
not transferred to the next configuration, thus simulating the execution of the decrement
on the counter $c_i$. In Figure \ref{fig:forwardzero}, we graphically depict the encoding of zero-test 
transitions.

\begin{wrapfigure}{r}{0.6\textwidth}
\vspace{-0.5cm}
\hspace{-0.3cm} 
\includegraphics{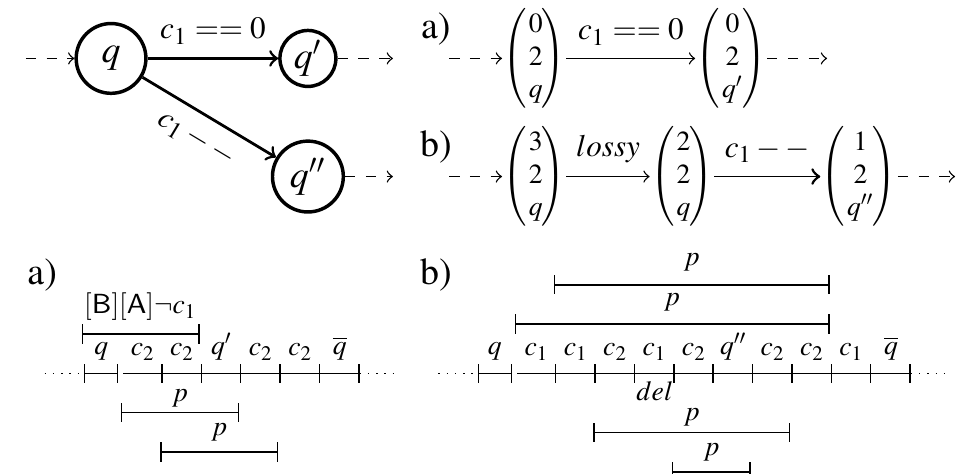}

\vspace{-0.1cm}
\caption{\hspace{-0.2cm} Encoding of a lossy computation  in $\AABB$: zero-test states.}\label{fig:forwardzero}
\vspace{-0.4cm}
\end{wrapfigure}

The next step is the correct transfer of all counter values from the current configuration to the next one
with the only exception of the 
$new/del$-labeled intervals (if any). What does ``correctly'' mean? According to 
the definition of the lossy semantics, a counter can be either transferred with its exact value or with a smaller one, that is, we only have to avoid unsupported increments of counter values, as for lower values we can always assume that a lossy transition has been fired (Figure \ref{fig:forwardinc} gives an example of such a behavior).
The transfer is done by means of a function that maps (the left endpoints of) $c_i$-labeled intervals of the current
configuration to (the left endpoints of) $c_i$-labeled intervals of the next one. Such a function is encoded by a proposition letter $p$.  
Notice that all the properties we dealt with so far, including those concerning $p$, can be expressed in $\ABB$. However, we still need to suitably constrain $p$ to guarantee that it behaves as expected. In particular, we must impose surjectivity to prevent unsupported increments of counter values from occuring. Such a property can be forced by the formula $\psi_{sur}= \boxG( \bigwedge_{i \in \{1, \ldots, k\}} c_i \wedge \neg new \rightarrow \diamondAbar  p)$, where modality $\diamondAbar$ plays an essential role.



Up to this point, the entire encoding has been done without exploiting any special feature of $\Sigma_{\spoiler}$ and $\Sigma_{\duplicator}$ brought by the finite synthesis context. The power of synthesis is needed to force injectivity,
thus turning lossy counter machines into standard ones.
Let us define the concrete instance of the synthesis problem we are interested in as the pair $(\varphi^{0-0}_{\cM}, \{p_{\spoiler}\})$, where $\varphi^{0-0}_{\cM} = \varphi_{\cM} \wedge \psi_{inj}$ and $\psi_{inj} = \boxG((p\wedge p_{\spoiler} \rightarrow \boxA (\neg \pi \rightarrow s)) \wedge (p \wedge \neg p_{\spoiler} \rightarrow \boxA( \neg \pi \rightarrow \neg s)))$ is the formula for injectivity.

Let us take a closer look at $\varphi^{0-0}_{\cM}$ to understand how it guarantees injectivity of the transfer function encoded by $p$. Let us assume $(\varphi^{0-0}_{\cM}, \{p_{\spoiler}\}) $ to be a positive instance of the finite synthesis problem. Then,  there exists a $\Sigma^T_{\spoiler}$-response strategy $S_{\diamond}$ such that 
for every $\rho_{\spoiler}$, a prefix 
$\rho[0\ldots n]$ of the response $\rho$ of $S_{\diamond}$ to  $\rho_{\spoiler}$ is a successful run, that is, ${\bM}_\rho[0\ldots n]$ satisfies $\varphi^{0-0}_{\cM}$. Assume by way of contradiction that for every successful run $\rho[0\ldots n]$ (for some natural number $n$), which is an $S_{\diamond}$ response to  $\rho_{\spoiler}$, the resulting model ${\bM}_\rho[0\ldots n] = (\{0, \ldots,n\}, \cV)$ is such that there exist three points $x < y < z \leq n$  with $p\in \cV(x,z)\cap \cV(y,z)$, thus violating injectivity of $p$.
We collect all these triples $(x,y,z)$ into a set $Defects(\rho)$. Without loss of generality, we can assume $n$ to be the minimum natural numbers such that ${\bM}_\rho[0\ldots n]$ satisfies $\varphi^{0-0}_{\cM}$. A  lexicographical order $\leq$ can be defined over the triples $(x,y,z)$ in $Defects(\rho)$, where $z$ is considered as the most significative component and $x$ as the least significant one. Let $(x,y,z)$ be the minimum element of $Defects(\rho)$ with respect to
$\leq$. First, we observe that $z$ cannot be equal to $n$, as, by construction, (i) $z$ is the left endpoint of a counter-labeled interval, (ii) any model for the formula is forced to end at the left endpoint of a state-labeled interval, and (iii) an interval cannot be both a state- and a counter-labeled interval (state- and counter-labeled intervals are mutually exclusive).
Hence, $z< n$.
Now, since ${\bM}_\rho[0\ldots n]$ satisfies $\psi_{inj}$, it immediately follows that either $p_{\spoiler} \in \cV(x,z)\cap \cV(y,z)$ or $p_{\spoiler} \not\in \cV(x,z)\cup \cV(y,z)$ (otherwise, both $s \in \cV(z,z+1)$ and $s \notin \cV(z,z+1)$). 
Without loss of generality, we assume that  $p_{\spoiler} \in \cV(x,z)\cap \cV(y,z)$ (the other case is completely symmetric). 
Let $i,j$ be the even indexes (player $\spoiler$ is playing at even positions) such that $\rho[i]=([x,z], \sigma_i)$ and $\rho[j]=([y,z], \sigma_j)$, respectively. Let $i<j$ (the opposite case is perfectly symmetric).
Since $\varphi$ must be satisfied by all $\rho_{\spoiler} \in \spruns(\varphi,\{p_{\spoiler}\})$, there exists a run $\rho'$ such that $\rho'[0\ldots j-1]=\rho[0\ldots j-1]$ (it is a run where spoiler behaves the same up to position $j-1$, and since
$S_{\diamond}$ is a function on the prefixes of runs, duplicator behaves the same as well) and $p_{\spoiler} \not\in \rho'[j]_{\sigma}$ (while $p_{\spoiler} \in \rho[j]_{\sigma}$).
By hypothesis, there exists $n'$ such that $\rho'[0\ldots n']$ is a model for $\varphi$. It clearly holds that $n'> z$; otherwise, minimality of $n$ would be violated, since $\rho'$ is equal to $\rho$ up to $z$. Now, from $p \in \rho'[i+1]_{\sigma}$ and $p_{\spoiler} \in \rho'[i]_{\sigma}$, it follows that $p \not\in \rho'[j+1]_{\sigma}$. 
%
Otherwise ($p \in \rho'[j+1]_{\sigma}$), a contradiction would occur when point $z+1$ is added, as $\psi_{inj}$
has forced player  $\diamond$ to put  both $s$  and $\neg s$ on the  interval $[z,z+1]$.
\begin{wrapfigure}{r}{0.3\textwidth}
\vspace{-0.5cm}
\hspace{-0.0cm} 
\includegraphics{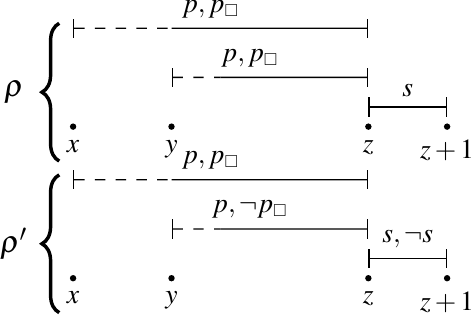}

\caption{\hspace{-0.2cm} Behavior of the formula $\psi_{inj}$.}\label{fig:forwardzero}
\vspace{-0.5cm}

\end{wrapfigure}
We can conclude that there exists a run $\rho'$ such that injectivity of $p$ 
is guaranteed up to the defect $(x,y,z)$, that is, the minimum element 
$(x',y',z')$ in $Defects(\rho')$ is greater than $(x,y,z)$ according to the above-defined lexicographical order.
Now, we can apply exactly the same argument we use for $\rho$ to $\rho'$, identifying a run $\rho''$
whose minimun defect $(x'',y'',z'')$ in $Defects(\rho'')$ is greater than $(x',y',z')$, and so on.
Let $\rho_{\omega}$ be the limit run. It holds that  
$\rho_{\omega}$ is still a response of $S_{\duplicator}$ to some $\rho_{\spoiler}$ in $\spruns(\varphi, \{p_{\spoiler}\})$ where $p$ is injective (contradiction).
 
\begin{theorem}\label{thm:aabbundecidabilitymachine}
Let $\cM$ be a counter machine. $\cM$ is a positive instance of  the $0$-$0$
reachability problem if and only if the $\AABB$-formula $\varphi^{0-0}_{\cM}$ is a 
positive instance of the finite synthesis problem.
\end{theorem}

\begin{corollary}\label{thm:aabbundecidability}
The finite synthesis problem for  $\AABB$ is undecidable.
\end{corollary}

As we already pointed out, the modality $\diamondAbar$ comes into play in the specification of
surjectivity only. Hence, if we drop the formula $\varphi_{sur}$, we can not force surjectivity 
anynore, but we can still impose injectivity (by exploiting the power of synthesis) in the smaller 
fragment $\ABB$. Then, by making use (with minor modifications) of a previous result of ours \cite{montanari2013adding}, we can reduce the reachability problem for lossy counter machines 
to the satisfiability problem for  $\ABBsim$ over finite linear orders. The main difference from the 
previous reduction is that computations are encoded backwards, that is, the encoding starts from
the final configuration and  (following the time line) it reaches the initial one.

\begin{wrapfigure}{l}{0.6\textwidth}
\vspace{-0.4cm}
\hspace{-0.3cm} 
\includegraphics{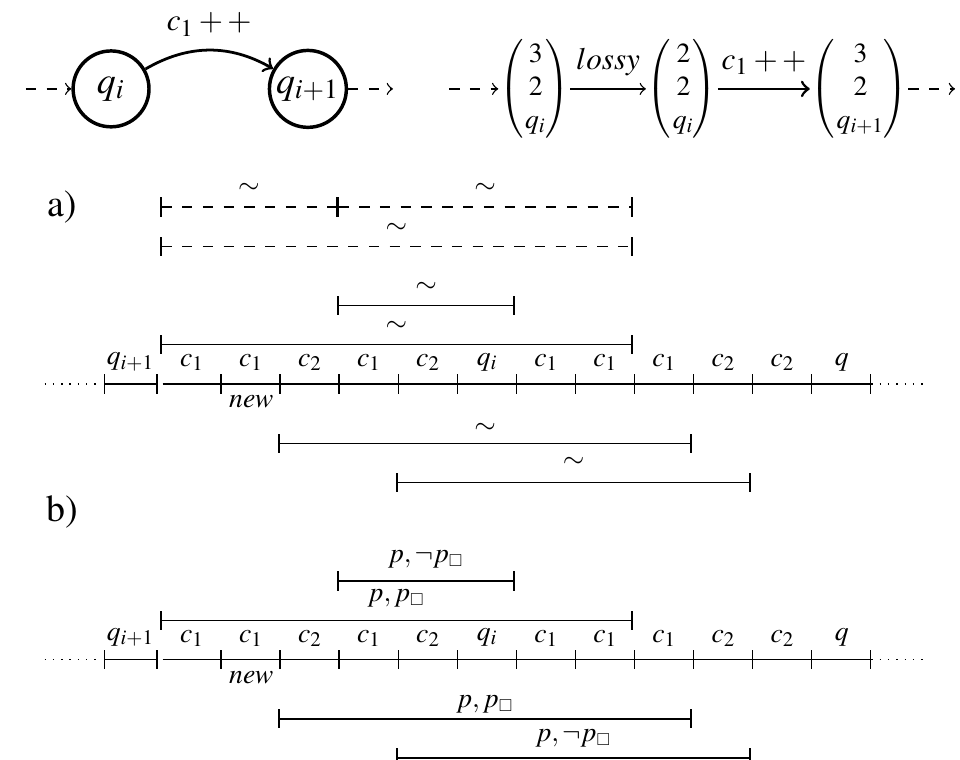}
\caption{\hspace{-0.3cm} Encoding of a lossy computation in $\ABBsim$ satisfiability (a) and
$\ABB$ synthesis (b): incrementing states.}\label{fig:backwardinc}
\vspace{-0.3cm}
\end{wrapfigure}

It can be easily checked that, in the synthesis setting, in order to express the lossy behavior, the transfer function 
must be injective. The absence of modality $\diamondAbar$ makes it impossible to provide an upper bound to
the value of a counter in the next configuration along the time line, which is the previous configuration in the computation of the counter machine. Let us consider Figure \ref{fig:backwardinc} (a). The second $c_i$-labeled interval of the configuration starting at $q_i$ does not begin any $\sim$-interval. However, since computations are encoded backwards, 
such a situation can be simulated by introducing a lossy transition. Injectivity can be forced by exploiting the equivalence relation (proposition letter $\sim$).
The formula $\varphi^{\sim}_{inj}=\boxG((\sim \wedge \neg \pi) \rightarrow \bigvee_{q \in Q} \diamondB\diamondA q )$ 
states that any $\sim$-interval, which is not a point-interval, must cross at least one state-interval. It immediately follows that points belonging to the same configuration must belong to different equivalence classes.  Suppose that the endpoints of a $\sim$-interval belong to the same $\sim$ class, against $\varphi^{\sim}_{inj}$. Transitivity of $\sim$ can then be exploited to violate injectivity, as shown by the dashed intervals in Figure \ref{fig:backwardinc} (a).
Injectivity can be forced in $\ABB$, which does not include the proposition letter $\sim$, by using the expressive power of synthesis. Thanks to formula $\psi_{inj}$, we can indeed mimic the behavior of $\sim$ by the combined behavior of 
$p$  and $p_{\spoiler}$ as shown in Figure  \ref{fig:backwardinc} (b). 
As an immediate consequence,  the non-primitive recursive hardness of the satisfiability problem for $\ABBsim$ over finite linear orders \cite{montanari2013adding} can be directly transferred to the finite synthesis problem for $\ABB$. 

\begin{theorem}\label{thm:aabhard}
The finite synthesis problem for $\ABB$ is Non-Primitive Recursive hard.
\end{theorem}

Similarly, the undecidability of the $\bbN$-synthesis problem for  $\ABB$ can be derived from the undecidability of the satisfiability problem of  $\ABBsim$ over linear orders isomorphic to $\bbN$ \cite{montanari2013adding}. 

\begin{theorem}\label{thm:aabundecidability}
The $\bbN$-synthesis problem for  $\ABB$ is undecidable.
\end{theorem}



\section{Decidability of $\ABBsim$ over finite linear orders}

We conclude the paper by showing that the synthesis problem for $\ABBsim$ over finite linear orders is decidable. To this end, we introduce some basic terminology, notations, and definitions.

Let $\bM=(\bbD, \cV)$ be an interval structure.
We associate with each interval $I\in\bbI_\bbD$ its \emph{$\varphi$-type} 
$\type_\bM^\varphi(I)$, defined as the set of all formulas $\psi\in\eclosure(\varphi)$ 
such that $\bM,I\sat\psi$ (when no confusion arises, we omit the parameters $\bM$ and 
$\varphi$). 
A particular role will be played by those types $F$ that contains the sub-formula
$\boxB\false$, which are necessarily associated with singleton intervals.
When no interval structure is given, we can still try to capture the concept of type 
by means of a maximal ``locally consistent'' subset of $\eclosure(\varphi)$.
Formally, we call \emph{$\varphi$-atom} any set $F\subseteq\eclosure(\varphi)$ such 
that (i) $\psi\in F$ iff $\neg\psi\nin F$, for all $\psi\in\eclosure(\varphi)$, (ii) $\psi\in F$ 
iff $\psi_1 \in F$ or $\psi_2 \in F$, for all $\psi=\psi_1 \vel \psi_2 \in\eclosure(\varphi)$, 
(iii) if $\boxB\false\in F$ and $\psi\in F$, then $\diamondA\psi\in F$, 
for all $\psi\in\closure(\varphi)$, (iv) if $\boxB\false\in F$ and $\diamondA\psi\in F$, 
then $\psi\in F$ or $\diamondBbar\psi\in F$, for all $\psi\in\closure(\varphi)$, and
(v) if $\boxB\false\in F$, then $\sim \in F$.
We call \emph{$\pi$-atoms} those atoms that contain the formula $\boxB\false$,
which are thus candidate types for singleton intervals.
We denote by $\atoms(\varphi)$ the set of all $\varphi$-atoms. 

Given an atom $F$ and a relation $R\in\{A,B,\bar{B}\}$, 
we let $\req_R(F)$ be the set of \emph{requests of $F$ along direction $R$}, 
namely, the formulas $\psi\in\closure(\varphi)$ such that $\langle \mathsf{R} \rangle \psi \in F$. 
Similarly, we let $\obs(F)$ be the set of \emph{observables of $F$}, namely, the 
formulas $\psi\in F\cap\closure(\varphi)$ -- intuitively, the observables of $F$ 
are those formulas $\psi\in F$ that fulfill requests of the form $\langle \mathsf{R} \rangle \psi$ 
from other atoms. Note that, for all $\pi$-atoms $F$, we have 
$\req_A(F) = \obs(F) \cup \req_{\bar{B}}(F)$.

It is well known that formulas of interval temporal logics can be equivalently interpreted over 
the so-called compass structures \cite{JLOGC::Venema1991}. These structures can be 
seen as two-dimensional spaces in which points are labelled with complete 
logical types (atoms). Such an alternative interpretation exploits the 
existence of a natural bijection between the intervals $I=[x,y]$ over a 
temporal domain $\bbD$ and the points $p=(x,y)$ in the $\bbD\times\bbD$ 
grid such that $x\le y$. It is useful to introduce a \emph{dummy atom}
$\emptyset$, distinct from all other atoms, and to assume that it labels 
all and only the points $(x,y)$ such that $x>y$, which do not 
correspond to intervals. Conventionally, we assume $\obs(\emptyset)=\emptyset$ 
and $\req_R(\emptyset)=\emptyset$, for $R\in\{A,B,\bar{B}\}$.

Formally, a \emph{compass $\varphi$-structure over a linear order $\bbD$} 
is a labeled grid $\cG=(\bbD\times\bbD,\tau)$, where the function 
$\tau:\bbD\times\bbD\then\atoms(\varphi)\uplus\{\emptyset\}$ maps 
any point $(x,y)$ to either a $\varphi$-atom (if $x\le y$) or the 
dummy atom $\emptyset$ (if $x>y$). 
Allen's relations over intervals have analogue relations over points $A,B,\bar{B}$ (by a
slight abuse of notation, we use the same letters 
for the corresponding relations over the points of a compass structure).
Thanks to such an interpretation, any interval structure $\bM$ can be converted 
to a compass one $\cG=(\bbD\times\bbD,\tau)$ by simply letting
$\tau(x,y)=\type([x,y])$ for all $x\le y\in\bbD$.
The converse, however, is not true in general, as the atoms associated with points 
in a compass structure may be inconsistent with respect to the underlying geometrical 
interpretation of Allen's relations. 
To ease a correspondence between interval and 
compass structures, we enforce suitable \emph{consistency conditions} on compass 
structures. First, we constrain each compass structures $\cG=(\bbD\times\bbD,\tau)$
to satisfy the following conditions on the 
special proposition letter $\sim$ (for the sake of readability 
we write $x\sim y$ in place of $\sim \in \tau(x,y)$):
(i)  for all $x\in D$, $x \sim x$;
(ii) for all $x<y<z$ in $D$,
     $x\sim y \wedge y \sim z \rightarrow x\sim z$, 
     $x\sim z \wedge y \sim z \rightarrow x\sim y$,
     and $x\sim z \wedge x \sim y \rightarrow y\sim z$.
Second, to guarantee the consistency of atoms associated with points, we introduce two  binary relations over them. 
Let $F$ and $G$ be two atoms: 
$$
\begin{array}{llll}
  F \depBbar G ~\text{iff}
  & \qquad\,
    F \depA G ~\text{iff} \\[0.5ex]
    \left\{
      \begin{array}{rcl}
        \req_{\bar{B}}(F) &\supseteq& \obs(G) \cup \req_{\bar{B}}(G) \\
        \req_B(G)         &\supseteq& \obs(F) \cup \req_B(F) \\
      \end{array}
    \right.
  & \qquad
    \begin{array}{l}
    \left\{
      \begin{array}{rcl}
        \req_A(F)         &=& \obs(G) \cup \req_B(G) \cup \req_{\bar{B}}(G)\!\!\!\! \\
      \end{array}
    \right.
    \\[2ex]
    \ 
    \end{array}
\end{array}
\vspace{-0.4cm}
$$

Note that the relation $\depBbar$ is transitive, while $\depA$ only satisfies
$\depA\circ\depBbar \;\subseteq\; \depA$. 
Observe also that, for all interval structures $\bM$ and all intervals $I,J$ in it,
if $I \mathrel{\bar{B}} J$ (resp., $I \mathrel{A} J$), then $\type_\bM(I) \depBbar \type_\bM(J)$ 
(resp., $\type_\bM(I) \depA \type_\bM(J)$).
Hereafter, we tacitly assume that every compass structure $\cG=(\bbD\times\bbD,\tau)$ 
satisfies analogous consistency properties with respect to its atoms, namely, 
for all points $p=(x,y)$ and $q=(x',y')$ in $\bbD\times\bbD$, 
with $x\le y$ and $x'\le y'$, if $p \mathrel{\bar{B}} q$ (resp., $p \mathrel{A} q$), 
then $\tau(p) \depBbar \tau(q)$ (resp., $\tau(p) \depA \tau(q)$).
In addition, we say that a request $\psi\in\req_R(\tau(p))$ of a point $p$ 
in a compass structure $\cG=(\bbD\times\bbD,\tau)$ is \emph{fulfilled} if there 
is another point $q$ such that $p \mathrel{R} q$ and $\psi\in\obs(\tau(q))$
-- in this case, we say that $q$ is a \emph{witness of fulfilment of $\psi$ from $p$}.
The compass structure $\cG$ is said to be \emph{globally fulfilling} if all requests 
of all its points are fulfilled.

We can now recall the standard correspondence between interval and compass 
structures (the proof is based on a simple induction on sub-formulas):

\begin{proposition}[\cite{montanari2013adding}]\label{prop:compass-structure}
Let $\varphi$ be an $\ABBsim$ formula. For every globally fulfilling compass 
structure $\cG=(\bbD\times\bbD,\tau)$, there is an interval structure 
$\bM=(\bbD,\cV)$ such that, for all
$x\le y\in\bbD$ and all $\psi\in\eclosure(\varphi)$, $\bM,[x,y]\sat\psi$ iff 
$\psi\in\tau(x,y)$.
\end{proposition}

In view of Proposition \ref{prop:compass-structure}, the satisfiability 
problem for an $\ABB\sim$ formula $\varphi$ reduces to the problem of 
deciding the existence of a compass $\tilde{\varphi}$-structure $\cG=(\bbD\times\bbD,\tau)$, 
with $\tilde{\varphi}=\diamondG\varphi$ ($\diamondG\varphi$ is a shorthand for
$\neg \boxG \neg \varphi$), that features the observable $\tilde{\varphi}$ 
at every point, that is, $\tilde{\varphi}\in\obs(\tau(x,y))$ for all $x\le y\in\bbD$.

It can be easily checked that, given an  $\ABB\sim$ formula $\varphi$, for every set 
$\sigma \subseteq \Sigma^T$, there exists at most one atom $F\in \atoms(\varphi)$
such that $F \cap \Sigma^T=\sigma$. Hence, for all $\ABB\sim$ formulas $\varphi$,
 we can define a (unique) partial function $f_{\varphi}: \cP(\Sigma^T) \rightarrow \atoms(\varphi)$ 
 such that for every set $\sigma \subseteq \Sigma^T$, $f_\varphi(\sigma)=F$ with $F 
 \cap \Sigma^T=\sigma$. 
By making use of the function $f_{\varphi}$, given an interval structure $\bM=(\bbD,\cV)$ for $\varphi$, 
we can define a corresponding compass structure $\cG_{\bM}= (\bbD\times\bbD, \tau)$ such that, for all 
$[x,y] \in \bbI(\bbD)$, $\tau(x,y)=f_{\varphi}(\cV([x,y]) \cup \{ \psi \in \Sigma^T\cap \eclosure(\varphi)\})$. 

\begin{lemma}\label{lem:closurefunction}
For every $\ABBsim$ formula $\varphi$ and 
interval structure $\bM$ for it, $\cG_{\bM}$ is a fulfilling compass structure for $\varphi$.
\end{lemma}

Let $\varphi$ be an $\ABB\sim$ formula,  $\bG_{\varphi}$ be the set of all finite compass structures for $\varphi$, and 
 $\Sigma^T_{\spoiler} \subseteq \Sigma^T$. We define a \emph{$\Sigma^T_{\spoiler}$-response tree} as a tuple $\cT=(V,E,\cL_V, \cL_E)$ where
\begin{itemize}
\item $(V,E)$ is a finite tree equipped with two labeling functions $\cL_V:V\rightarrow \Sigma^T  \setminus\Sigma^T_{\spoiler}$
and $\cL_E:E \rightarrow  \bbI(\bbN)\times \Sigma^T_{\spoiler}$ 
(we denote
the projection of $\cL_E$ on the first component  by $\cL_E|_I$); 
\item for each root-to-leaf path $\pi=(v_0,v'_0)\ldots(v_n,v'_n)$ in $\cT$, $\rho_{\pi}=\cL_E(v_0,v_0')(\cL_E|_I(v_0,v_0'), \cL_V(v_0'))\ldots$ $\cL_E(v_n,v_n')(\cL_E|_I(v_n,v_n'), \cL_V(v_n'))$ is a successful admissible run for $\varphi$ (hereafter, we denote the set of all runs  $\rho_{\pi}$ associated with root-to-leaf paths in $\cT$ by $\runs_{\cT}$);
\item for each run $\rho_{\spoiler}$ in $\spruns$, there is a path $\pi=(v_0,v'_0)\ldots(v_n,v'_n)$
such that $\rho_{\spoiler}[0\ldots n]=\cL_{E}(v_0,v'_0)\ldots$ $\cL_E(v_n,v'_n)$.
\end{itemize}

It can be easily checked that a $\Sigma^T_{\spoiler}$-response tree encodes some (successful) $\Sigma^T_{\spoiler}$-response strategy $S_{\diamond}$. Then, checking whether a formula $\varphi$  and a  set  $\Sigma^T_{\spoiler} \subseteq \Sigma^T$ are a positive instance of the finite synthesis problem amounts to check whether there exists a
$\Sigma_{\spoiler}$-response tree for $\varphi$.

\begin{theorem}\label{thm:response-tree}
Let  $\varphi$ be an $\ABBsim$ formula and $\Sigma^T_{\spoiler} \subseteq \Sigma^T$. Then, ($\varphi$, $\Sigma_{\spoiler} $) is a positive instance of the finite-synthesis problem if and only if there exists a  $\Sigma^{T}_{\spoiler}$-response tree  $\cT$ for  $\varphi$.
\end{theorem}

Unfortunately, this is not the end of the story. Every $\Sigma^T_{\spoiler}$-response tree  $\cS=(\cT, \cL_\pi)$ is finite 
by definition, but this is not sufficient to conclude that the finite synthesis problem is decidable. To this end, we must
provide a bound on the height of the tree depending on the size of $\varphi$.
Let $\cG=(\bbD\times\bbD,\tau)$ be a compass structure. For all $x,y \in D$, with $x\leq y$, we define the \emph{multi-set of atoms} $M(x,y)=\{F: \exists x' \in D (x'\leq y \wedge x'\sim x \wedge \tau(x',y)=F)\}$, where the number of copies of $F$ in $M(x,y)$, denoted by $|M(x,y)(F)|$, is equal to $|\{x' \in D : x'\leq y \wedge x'\sim x \wedge \tau(x',y)=F\}$.
Moreover, for all $y \in D$, we define the multiset-collection $\cM(y)$ as the multi-set of multi-sets of atoms such that, for each multiset of atoms $M$, $\cM(y)(M)=|\{ [x]_\sim: x\in D \wedge  x\leq y \wedge M(x,y)=M   \}|$. Finally, we define a partial order $\leq$ over the set of all multi-set collections as follows: for any pair of multi-set collections $\cM, \cM'$, $\cM \leq \cM'$ if and only if there exists an injective multi-set function $g \subseteq \cM \times \cM'$ such that, for each pair $(M,M')\in g$, $M \subseteq M'$ (an injective function $g\subseteq \cM \times \cM'$ between two multisets  is itself a multiset  such that $g|_1=\cM$ and $g|_2\subseteq \cM'$, where $|_i$ is simply the projection on the $i$-th component of a tuple). The following property of $\leq$ is not difficult to prove, but it is crucial for the decidability proof.

\begin{lemma}\label{lemma:wqo}
$\leq$ is a well-quasi-ordering (WQO) over multiset collections.
\end{lemma}

Let $\cG=(\bbD\times\bbD,\tau)$ be a compass structure for $\varphi$. We say that $\cG$ is \emph{minimal} 
if and only if for all $y<y'$ in $D$, $\cM(y)\not\leq \cM(y')$. Given an $\ABBsim$ formula $\varphi$
and a $\Sigma^T_{\spoiler}$-response tree  $\cT$ for it, we say that  $\cT$ is \emph{minimal} if  and only if for each 
$\rho \in \runs_{\cT} $, $\cG_{\rho_{\bM}}$ is a minimal compass structure. The following result allows us to restrict our attention to minimal $\Sigma^T_{\spoiler}$-response trees.

\begin{lemma}\label{lemma:contraction}
Let $\varphi$ be an $\ABBsim$ formula and $\Sigma^T_{\spoiler} $ be a finite set of its variables. 
Then, if there exists a $\Sigma^T_{\spoiler}$-response tree  $\cT$ for  $\varphi$, then there exists a minimal
$\Sigma^T_{\spoiler}$-response tree  $\cT'$ for $\varphi$.
\end{lemma}
\begin{proof}\emph{(sketch)}
Let  $\cT=(V,E,\cL_V, \cL_E)$ be a $\Sigma^T_{\spoiler}$-response tree for $\varphi$. Suppose that $\cT$ is not 
minimal. We show that there exists a smaller $\Sigma^T_{\spoiler}$-response tree for $\varphi$ which can be obtained 
by contracting one among the paths of $\cT$ that violate minimality. Notice that, in doing that, we prove that any defect (with respect to minimality) can be fixed by reducing the size of the tree in such a way that the resulting tree is still  
a $\Sigma^T_{\spoiler}$-response tree for $\varphi$.

Since $\cT$ is not minimal, there is a run $\rho \in \runs_{\cT} $ such that $\cG_{\rho_{\bM}}=(\bbD\times\bbD,\tau)$ is not a minimal compass structure. Then, there are $y<y'$ in $D$ such that $\cM(y)\leq \cM(y')$.
Let $g\subseteq \cM \times \cM'$ be the 
function from $\cM(y)$ to $ \cM(y')$, whose existence is guaranteed by definition of $\leq$. By the definition of the collections, injectivity of $g$ implies the existence of an injective function $f:\{0,\ldots,y\}\rightarrow \{0,\ldots,y'\}$ such that $\tau(x,y)=\tau(f(x),y')$ and  for each pair $0\leq x \leq x'\leq y$, $x\sim x'$ if and only if $f(x) \sim f(x')$.
Let $\pi=(v_0,v'_0)\ldots(v_n,v'_n)$ be a root-to-leaf path such that  $\rho=\cL_E(v_0,v_0'),(\cL_E|_I(v_0,v_0'),$ $\cL_V(v_0'))\ldots \cL_E(v_n,v_n')(\cL_E|_I(v_n,v_n'), \cL_V(v_n'))$ (the existence of such a path is guaranteed by 
the definition of $\Sigma_{\spoiler}$-response tree).
Given a node $v\in V$, we denote by  $\cT_v=(V_v, E_v, \cL_{E_v}, \cL_{V_v})$ the  sub-tree of $\cT$ rooted at $v$.
Let $v_i$ (resp., $v_j$) be a node in $\pi$ such that $i$ (resp., $j$) is the minimum index for which the interval  $[x,y'']=\cL_E|_I(v_{i'},v_{i'+1}')$ (resp., $[x,y'']=\cL_E|_I(v_{j'},v_{j'+1}')$) satisfies $y''>y$ (resp., $y''>y'$), for all $i'\geq i$ (resp., $j'\geq j$). 

Let 
$V'_{v_j}=\{ v \in V_{v_j}: \exists \pi'=(v_0,v_0')\ldots(v_m, v'_m) \text{ in $\cT_{v_j}$  s.t.  } v_0=v_j \wedge v'_m=v 
\wedge \forall 0\leq i \leq m \text{ if }  \cL_{E_{v_j}}(v_i,v'_i)|_I=[x,y''] \text{ then } x>y' \vee x\in \img(f) \}$.
The set  $V'_{v_j}$ collects all and  only the nodes reachable from $v_j$ through a path of edges which feature only intervals $[x,y'']$ such that either $x\geq y'$ (that is, the point has been introduced after $y'$) or $x \in \img(f)$.
Moreover, let 
$E'_{v_j}= \{(v,v')\in E_{v_j}: v,v'\in V'_{v_j}\cup \{v_j\}\}$ be the set of edges restricted to the set $V'_{v_j}$ and let $\cL'_{V_{v_j}}(v)=\cL_{V_{v_j}}(v)$ for all $v \in V'_{v_j}$. 
Finally, let $\Delta=y'-y$. We define $\cL'_{E_{v_j}}$ in such a way that, for each $(v,v')\in E'_{v_j}$,
if $(\cL_{E_{v_j}}(v,v')=)\cL_E(v,v')=([x,y''], \sigma_{\spoiler})$, then $\cL'_{E_{v_j}}(v,v')=([x',y'-\Delta],\sigma_{\spoiler})$, where  $x'=f^{-1}(x)$ if $x\leq y'$ (in such a case, by construction, $x\in\img(f)$ and, since $f$ is injective, it can be inverted on $x$) or $x'=x-\Delta$ otherwise.  

We complete the construction by replacing (in $\cT$) the subtree $\cT_{v_i}$ by the subtree $\cT'_{v_j}=(V'_{v_j}\cup \{v_i \},  E'_{v_j} \setminus \{ (v_j, v) \in E \} \cup \{ (v_i,v): \exists (v_j,v)\in E \} ,\cL''_{V_{v_j}} , \cL''_{E_{v_j}})$,
where $\cL''_{V_{v_j}}(v)=\cL'_{V_{v_j}}(v)$
for each $v\in V'_{v_j}$, $\cL''_{V_{v_j}}(v_i)=\cL_V(v_i)$, 
 $\cL''_{E_{v_j}}(v,v')=\cL'_{E_{v_j}}(v,v')$
for each $(v,v')\in E'_{v_j}$ with $v,v'\in V'_{v_j}$, 
and $\cL''_{E_{v_j}}(v_i, v)=\cL'_{E_{v_j}}(v_j, v)$ 
for each $(v_j,v)\in E'_{v_j}$ with $v\in V'_{v_j}$.
It is possible to prove (by induction) that the  tree $\cT'$ obtained from such a contraction operation on subtrees is still a $\Sigma^T_{\spoiler}$-response tree for $\varphi$.
\end{proof}
 

This proof provides the necessary insights for devising a decision procedure to establish whether or not $(\varphi, \Sigma_{\spoiler})$ is a positive instance of the finite synthesis problem. Such a procedure visits a (candidate) $\Sigma_{\spoiler}$-response tree $\cT$ for $\varphi$ in a breadth-first fashion. At each step,  the number of total edges from level $l$ to level $l+1$ of the tree is finite, and thus their labeling $\cL_E$ as well as labeling $\cL_V$ for the nodes at level $l+1$ can be nondeterministically guessed. The procedure returns success if it finds a  $\Sigma_{\spoiler}$-response tree for $\varphi$; it returns failure  (that is, the generated candidate tree is not minimal) if it either introduces some local inconsistency or it produces a path such that there exist two coordinates $y_i<y_j$ with $\cM(y_i) \leq \cM(y_j)$. Being $\leq$ is a WQO guarantees that a path cannot be arbitrarily long. 

\begin{theorem}\label{thm:decidability}
The finite synthesis problem for $\ABBsim$ is decidable.
\end{theorem}

\section{Conclusion}\label{sec:conclusion}
In this paper, we explored the synthesis problem for meaningful fragments of $HS$ in the presence of an equivalence relation over points. On the negative side, we proved that the computational complexity of the synthesis problem is generally worse than that of the corresponding satisfiability problem (from elementary/decidable to nonelementary/undecidable). On the positive side, we showed that the increase in expressiveness makes it possible to capture new interesting temporal conditions.

\bibliographystyle{eptcs}
\bibliography{biblio}
\end{document}